\documentclass{elsart}
\usepackage{dcolumn,amsthm}
\usepackage{bm}
\usepackage{lmodern}
\usepackage{dsfont}
\usepackage{mathrsfs}
\usepackage{mathtools}
\usepackage{mathbbol}
\usepackage{hyperref}
\usepackage{xcolor}

\newcommand{\reals}{\mathds{R}}
\newcommand{\complex}{\mathds{C}}

\begin{document}
\theoremstyle{plain}
\newtheorem{theorem}{Theorem}
\newtheorem{corol}{Corollary}
\theoremstyle{definition}
\newtheorem{definition}{Definition}
\newtheorem{example}{\small Example}
\newtheorem{lema}{Lemma}
\newtheorem{propos}{Proposition}

\begin{frontmatter}
\title{Swift chiral quantum walks}
\author{Massimo Frigerio ${}^{1,2}$,  Matteo G. A. Paris ${}^{1,2}$}
\address{
${}^{1}$ Quantum Technology Lab, Dipartimento di Fisica {\em Aldo Pontremoli}, Universit\`a degli 
Studi di Milano, I-20133 Milano, Italy \\
${}^{2}$ INFN, Sezione di Milano, I-20133 Milano, Italy 
}
\date{\today}
\begin{abstract}
A continuous-time quantum walk (CTQW) is {\em sedentary} if the return probability in the starting vertex is close to one at all times. Recent results imply that, when starting from a maximal degree vertex, the CTQW dynamics generated by the Laplacian and adjacency matrices are typically sedentary. In this paper, we show that the addition of appropriate complex phases to the edges of the graph, defining a \emph{chiral} CTQW, can cure sedentarity and lead to \emph{swift chiral quantum walks} of the adjacency type, which bring the return probability to zero in the shortest time possible. We also provide a no-go theorem for swift chiral CTQWs of the Laplacian type. Our results provide one of the first, general characterization of tasks that can and cannot be achieved with chiral CTQWs.
\end{abstract}


\end{frontmatter}
\section{Introduction}
In the physical modeling of quantum discrete structures, the notion of continuous-time quantum walks (CTQWs) is widely adopted as a convenient mathematical construct \cite{qw1,qw2,qw3,qw4,qw5,qw6,qw7}. Their applications range from quantum computation \cite{childs2009universal} and quantum search \cite{shenvi2003quantum,childs2004spatial} to the study of perfect state transfer \cite{alvir2016perfect,christandl2004perfect,kendon2011perfect,mathcycles} and quantum transport \cite{beggi2018probing,agliari2008dynamics,salimi2010continuous,darazs2014transport,li2020quantum,rai2008transport,blumen2006coherent,mulken2007quantum,xu2008coherent,yalouz2018continuous,mulken2007survival,agliari2010continuous,cav22,razzoli11}.
Let $\mathcal{G} = ( \mathcal{V} , \mathcal{E})$ be a simple, connected, unweighted and undirected graph with vertex set $\mathcal{V}$ and edge set $\mathcal{E}$. Let $\vert \mathcal{V} \vert = N$ be the order of $\mathcal{G}$ and $\mathbf{A}$ its adjacency matrix, once a numbering of the vertices has been assigned, i.e. $\mathbf{A}_{jk} = 1$ if vertices $j$ and $k$ are connected, while all the other entries (including the diagonal ones) are null. We will denote by $\mathbf{L} = \mathbf{D} - \mathbf{A}$ the Laplacian of $\mathcal{G}$, where $\mathbf{D}$ is the diagonal matrix encoding the (positive) degrees of the vertices of $\mathcal{G}$ \cite{godsilbook,graphspectra}. In its simplest form, a quantum walk on $\mathcal{G}$ is described by a Hilbert space isomorphic to $\complex^{N}$ with a reference basis formed by vectors $\underline{e}_{v}$ that are the characteristic vectors of the vertices of $\mathcal{G}$ (the vertex basis), and by a one-parameter subgroup of the unitary group on $\complex^{N}$ generated by either $\mathbf{A}$ or $\mathbf{L}$ \cite{wong2016laplacian} and effecting the time evolution of the vectors in the Hilbert space. More concretely, given an initial vector $\underline{\psi}_{0} = \underline{\psi}(0)$, its evolution will be governed by the Schr\"{o}dinger equation \cite{sakurai2011san}:
\begin{equation}
    i \dfrac{d}{dt} \underline{\psi} (t) \ = \ \mathbf{H} \  \underline{\psi}(t)
\end{equation}
where $\mathbf{H}$ can be $\mathbf{A}$ or $\mathbf{L}$, and whose solution is simply:
\begin{equation}
    \underline{\psi}(t) \ = \ e^{- i \mathbf{H} t } \underline{\psi}_{0}
\end{equation}
Recently, more generic hermitian generators, still reflecting the connectivity of $\mathcal{G}$, have been considered \cite{cqw1,cqw2,frigerio21a,frigerio21b}. Specifically, a \emph{chiral adjacency matrix} $\tilde{\mathbf{A}}$ on $\mathcal{G}$ is an $N \times N$ hermitian matrix such that $\mathbf{A}_{jk} =1 \implies \tilde{\mathbf{A}}_{jk} = e^{i \theta_{jk}}$ with $\theta_{jk} \in [0, 2\pi )$ and $\mathbf{A}_{jk}= 0 \implies \tilde{\mathbf{A}}_{jk} = 0$ ($\theta_{jk} = - \theta_{kj}$ is required by the hermiticity condition). In other words, a generic chiral adjacency matrix $\tilde{\mathbf{A}}$ is obtained by replacing each nonzero entry of $\mathbf{A}$ with a generic complex phase and enforcing that the resulting matrix is hermitian. A \emph{chiral Laplacian} on $\mathcal{G}$ is an hermitian matrix $\tilde{\mathbf{L}} = \mathbf{D} - \tilde{\mathbf{A}}$ where $\tilde{\mathbf{A}}$ is a chiral adjacency matrix. Finally, a \emph{chiral Hamiltonian} on $\mathcal{G}$ is an hermitian matrix of the form $\mathbf{H} = \tilde{\mathbf{D}} + \tilde{\mathbf{A}}$ where $\tilde{\mathbf{D}}$ is a generic, real diagonal matrix and $\tilde{\mathbf{A}}$ again a chiral adjacency matrix. We will say that a chiral quantum walk is \emph{of the adjacency type} (resp. \emph{of the Laplacian type}) if it is generated by a chiral adjacency matrix (resp. a chiral Laplacian). Among others, the typically relevant quantities for these systems are transport probabilities between two vertices $j$ and $k$, defined by:
\begin{equation}
    p^{\mathbf{H}}_{j \to k} (t) \ = \ \vert \underline{e}_{k}^{\dagger} e^{-i t \mathbf{H}} \underline{e}_{j} \vert^{2}
\end{equation}
where now $\mathbf{H}$ can be a generic chiral Hamiltonian on the graph and $\vert \cdot \vert$ is the modulus of a complex number. Given a vertex $v \in \mathcal{G}$ and a numbering of the vertices, its characteristic vector $\underline{e}_{v}$ is the vector in the canonical basis of $\complex^{N}$ whose only nonzero entry is in correspondence to the number associated with $v$.

\subsection{Sedentarity in Laplacian and adjacency continuous-time quantum walks}
To set the stage for our results, it is useful to introduce some nomenclature for graphs with maximally connected vertices: a \emph{cone} on a graph $\mathcal{G}'$ is a graph $\mathcal{G}$ obtained by adding a vertex $v$ to the vertex set of $\mathcal{G}'$ and connecting it to every vertex in $\mathcal{G}'$.  
Recently, a number of results for both Laplacian and adjacency type CTQWs \cite{godsilsedent,razzoli22} indicate that the evolution starting from maximally connected vertices can be plagued by a sedentary behaviour, meaning that the return probability to the initial vertex stays close to $1$ at all times if its degree is sufficiently large. The maximally connected vertex $v$ will be called the \emph{apex} of the cone. For Laplacian quantum walks this is always the case, since their evolution from a maximally connected vertex in any graph with $N+1$ vertices is the same as the one from the center (labelled by $1$) of an $N$-pointed star graph, whose return probability is \cite{Xu_2009}:
\begin{equation}
\label{eq:Lthm}
    p^{L}_{1 \to 1} (t) \ = \ 1 - 2 \frac{N-2}{(N-1)^2} \left[ 1 - \cos ( N t ) \right]
\end{equation}
and it is clearly sedentary, since $p^{L}_{1 \to 1} (t) \rightarrow 1 \ \forall t \in \reals$ as $N \to +\infty$. This general result was proved in \cite{razzoli22}. 

Concerning adjacency CTQWs, instead, the situation is currently less clear. The dynamics starting from the central vertex of the $N$-pointed star graph in this case is not sedentary, nor it is the one which starts from the center of the wheel graph. Moreover, each adjacency CTQW starting from a maximally connected vertex could yield a different evolution, depending upon the connectivity of the remaining vertices. However, it was shown in \cite{godsilsedent} that adjacency quantum walks from the apex of cones over several infinite families of graphs are sedentary. Importantly, this is the case when the cone is over $m$-regular graphs with $N$ vertices and the degree $m$ grows faster than linearly with $N$. 

\par{In this paper, we show that the addition of appropriate complex phases to the edges of several graphs eliminates sedentarity and allows one to obtain \emph{swift chiral quantum walks} of the adjacency type, where the returning probability goes to zero in the shortest time possible. We also provide a no-go theorem for swift chiral CTQWs of the Laplacian type. Overall, our results provide one of the first, general characterization of tasks that can and cannot be achieved with chiral CTQWs and they hint at several challenging graph-theoretical questions concerning the existence of swift phases configurations.}

\section{Results for adjacency quantum walks}
Let us consider a quantum walk starting at the apex of a cone graph $\mathcal{G}$ over an $m$-regular graph of $N$ vertices, whose state vector will be labelled $\underline{e}_{1}$. Since this vertex is maximally connected, the adjacency matrix $\mathbf{A}$ of the full graph will have the following structure:
\begin{equation}
\label{eq:AMgraph}
    \mathbf{A} \ = \ \left( \begin{array}{c|cccc}
    0  & 1 & 1 & 1 & \dots \\
    \hline 
    1 & & & & \\
    1 & & & & \\ 
    1 & & &  \mathbf{A'}  & \\
    \vdots & & & &  \\  \end{array} \right)
\end{equation}
where $\mathbf{A'}$ is the adjacency matrix of the $m$-regular graph $\mathcal{G}'$. Since each vertex has exactly $m$ neighbours, calling $\underline{1}_{N}$ the vectors of all ones in the subspace of $\mathcal{G}'$ we have:
\begin{equation}
\label{eq:A'reg}
    \mathbf{A'} \underline{1}_{N}  \ = \ m \underline{1}_{N} 
\end{equation}
But notice that $\mathbf{A}$ maps $\underline{e}_{1}$ to the vector $\underline{e}'= (0,1,1, \cdots)^{T}$ which coincides with $\underline{1}_{N}$ on $\mathbf{A'}$ and is orthogonal to $\underline{e}_{1}$. Moreover, it is clear from Eq.(\ref{eq:A'reg}) that:
\begin{equation}
\mathbf{A} \underline{e}' \ = \ N \underline{e}_{1} \ + \ m \underline{e}'
\end{equation}
Therefore, calling $\underline{e}_{E} = \frac{1}{\sqrt{N}} \underline{e}'$ the normalized associated with $\underline{e}'$, $\mathbf{A}$ reduces to a simple $2 \times 2$ block in the space spanned by $\{ \underline{e}_{1}, \underline{e}_{E} \}$ \cite{qwreduct}:
\begin{equation}
    \mathbf{A}_{\mathrm{red}} = \left( \begin{array}{cc} 0 & \sqrt{N} \\
    \sqrt{N} & m \end{array} \right)
\end{equation}
This is easily diagonalized, and the resulting return probability starting from $\underline{e}_{1}$ is given by:
\begin{equation}
\label{eq:univadjreg}
    p^{\mathbf{A}}_{1 \to 1} (t) \ = \ 1 -\frac{2 N}{m^2 + 4 N} \left[ 1 - \cos( \sqrt{m^2 + 4 N} t) \right]
\end{equation}
whereas, with probability $1 - p^{\mathbf{A}}_{1 \to 1}(t)$ the walker will be in the state $\underline{e}_{E}$ orthogonal to $\underline{e}_{1}$, from which all the other site-to-site transition probabilities follow.\\

Eq.(\ref{eq:univadjreg}) can be regarded as a partial result on the adjacency version of the general theorem (\ref{eq:Lthm}) for Laplacian quantum walks. It states that the dynamics starting from the apex of a cone over \emph{any} $m$-regular graph is fully specified by $m$ and the number of vertices of the underlying graph $N$ (notice that the full graph has $N+1$ vertices instead). If we define a family of regular graphs such that the degree scales as $m \sim O(N^{\textcolor{red}{\frac{1}{2}}+ \alpha})$ with $\alpha >0$, it follows that $p^{A}_{1 \to 1}$ stays very close to $1$ at all times, for large enough $N$. This construction leads to a result originally presented in \cite{godsilsedent}. \\

\section{Swift chiral quantum walks}
Now suppose that we can attach complex phases to the edges of the $m$-regular graph $\mathcal{G}'$ such that the resulting chiral adjacency matrix $\tilde{\mathbf{A}}$, which would substitute $\mathbf{A'}$ in Eq.(\ref{eq:AMgraph}), fulfils the following condition:
\begin{equation}
\label{eq:condphases}
    \tilde{\mathbf{A}} \underline{1}_{N} \ = \ 0
\end{equation}
that is to say, the phases conjure in such a way that the sum of each row of $\tilde{\mathbf{A}}$ is zero. Retracing the steps to arrive at Eq.(\ref{eq:univadjreg}), we find the following reduced Hamiltonian for the chiral QW on the cone in the subspace spanned by $\{ \underline{e}_{1} ,\underline{e}_{E} \}$:
\begin{equation}
\label{eq:Hswift}
     \mathbf{H}_{\mathrm{red}} = \left( \begin{array}{cc} 0 & \sqrt{N} \\
    \sqrt{N} & 0 \end{array} \right)
\end{equation}
from which the return probability to the starting, maximally connected vertex is calculated to be:
\begin{equation}
\label{eq:swift}
     p^{\tilde{\mathbf{A}}}_{1 \to 1} (t) \ = \  \cos^2 ( \sqrt{N} t)
\end{equation}
Notice that this is the same return probability that would result from the adjacency quantum walk starting from the center of an $N$-pointed star graph, which cannot support chiral quantum walks since it is a tree graph. This equation motivates the following definition: \\

\begin{definition}
Let $\mathcal{G} = ( \mathcal{V} , \mathcal{E} )$ be a simple, connected and undirected graph, let $v \in \mathcal{V}$ be a vertex of $\mathcal{G}$ and denote by $N$ the cardinality of the set of neighbours of $v$ in $\mathcal{G}$. 
Given a chiral Hamiltonian $\mathbf{H}$ on $\mathcal{G}$, it is said to generate a \emph{\textbf{swift chiral quantum walk}} starting from \textcolor{red}{$v$ in} $\mathcal{G}$ if:
\begin{equation}
\label{eq:condswift}
   \forall t \geq 0 : \ \  \vert \underline{e}_{v}^\dagger  e^{ - i \mathbf{H} t } \underline{e}_{v} \vert^{2} \ = \ \cos^{2} \left( \sqrt{N} t \right)
\end{equation}
\end{definition}

\par{ Therefore, whenever it is possible to attach phases to the edges such that Eq.(\ref{eq:condphases}) is fulfilled, we obtain a similar result to the one which is valid for the Laplacian quantum walks, in the sense that all the adjacency-type \emph{swift} chiral quantum walks starting from a maximally connected vertex have the same dynamics as that of the adjacency quantum walk on the star graph. Notice, however, that Laplacian quantum walks tend to stay on the initial vertex, if this has maximal degree, whereas swift quantum walks abandon it fast. In the following, we will formalize this observation and the nomenclature of \emph{swift chiral quantum walks} with rigorous and general results.}

\par{\textbf{\emph{Remark.}} If $\mathcal{G}'$ has a vertex of degree $1$, then Eq.(\ref{eq:condphases}) cannot be fulfilled by any assignment of phases on its edges.   }\\

\par
{\begin{theorem}
Let $\mathcal{G}$ be a graph and let $v$ be a vertex of $\mathcal{G}$ with degree $N$. Out of all the chiral quantum walks starting from $v$ in $\mathcal{G}$ (with generic diagonal entries of the chiral Hamiltonian), a \emph{swift chiral quantum walk}, when it exists, is the fastest at leaving $v$, in the sense that it evolves to a state orthogonal to $v$ in a time $\tau_{s} = \frac{\pi}{2 \sqrt{N}}$ which achieves the \textbf{shortest} quantum speed limit among all chiral evolutions starting from $v$ on $\mathcal{G}$.
\end{theorem}

\begin{proof}
The speed limit to completely leave $v$ is the speed limit to go to a generic state orthogonal to $\underline{e}_{v}$. The generic formula of the quantum speed limit for orthogonal final state is \cite{QSL03,QSL03a,QSL04,QSL17}:
\begin{equation}
\label{eq:QSL1}
    \tau_{QSL}   :=  \max \left\{ \frac{  \pi }{ 2 \Delta \mathbf{H}}  \ , \ \frac{\pi }{ 2 ( \langle \mathbf{H} \rangle - E_{0} ) }        \right\} 
\end{equation}
where $\Delta \mathbf{H} = \sqrt{\langle \mathbf{H}^{2} \rangle - \langle \mathbf{H} \rangle^{2}}$ is the standard deviation of energy on the orbit of $\underline{e}_{v}$, $\langle \mathbf{H} \rangle = \underline{e}_{v}^{\dagger} \mathbf{H} \underline{e}_{v}$ its average energy and $E_{0}$ is the ground state's energy. Notice that the quantum speed limit is defined for a \emph{fixed} Hamiltonian, and we are looking for the \emph{smallest} quantum speed limit over all chiral Hamiltonians on $\mathcal{G}$.

But for any localized state $\vert v \rangle$ with degree $N$ and any chiral Hamiltonian $\mathbf{H}$ on $\mathcal{G}$, it is always true that $\Delta \mathbf{H} = \sqrt{N}$. Indeed, let $ \underline{e}_{v}^{\dagger} \mathbf{H} \underline{e}_{v} = d$, then $\underline{e}_{v}^{\dagger} \mathbf{H}^{2} \underline{e}_{v} = \underline{h}_{v}^{\dagger} \underline{h}_{v} = d^2 +  N$ where $\underline{h}_{v}$ is the $v$-th row of $\mathbf{H}$, whose $v$-th entry is $d$ and all the other non-zero entries are precisely $N$ complex phases. Therefore we can rewrite:
\begin{equation}
\label{eq:QSL2}
    \tau_{QSL}   :=  \max \left\{ \frac{  \pi }{ 2 \sqrt{N} } \ , \ \frac{\pi }{ 2 (  - E_{0} ) }        \right\}  \ = \  \max \left\{ \tau_{s} \ , \ \frac{\pi }{ 2 (  - E_{0} ) }        \right\}
\end{equation}
 To conclude the proof, let us suppose that there is a better chiral Hamiltonian for which the second expression of Eq.(\ref{eq:QSL2}) applies, which, because of the $\max$, implies that $\tau_{QSL} > \tau_{s}$. But since we supposed that a swift chiral quantum walk exists, we know that $\tau_{s}$ can be achieved by the swift chiral Hamiltonian, which shows that the second expression in Eq.(\ref{eq:QSL2}) is never the best quantum speed limit whenever a swift chiral quantum walk exists. As a byproduct, we deduce that for swift chiral quantum walks the ground state energy fulfils $E_{0} \leq - \sqrt{N}$.
\end{proof} }

\textbf{\emph{Remark.}} It is not immediate to conclude that any chiral quantum walk which achieves this limit $\tau_{s}$ must be a swift chiral quantum walk since, a priori, the specific form of the return probability need not be given by Eq.(\ref{eq:condswift}). \\

It is a very interesting question to ask when is it possible to construct a swift chiral quantum walk (of the adjacency type) starting from maximally connected vertices in any graph $\mathcal{G}$ having no vertex with degree $2$ or, equivalently, to find a chiral adjacency matrix fulfilling Eq.(\ref{eq:condphases}) on any graph $\mathcal{G}'$ having no vertex of degree $1$. We shall introduce the following definition: \\

\begin{definition}
\label{def:swiftconfig}
Let $\mathcal{G}'$ be a graph of order $N$ and suppose there exists a chiral Hamiltonian $\mathbf{H}$ on $\mathcal{G}'$ with eigenvector $\mathbf{1}_{N}$ and eigenvalue $0$ (where $\mathbf{1}_{N} \in \reals^{N}$ is the vector of all ones) or, equivalently, such that the sum of each row of $\mathbf{H}$ is zero. We will say that the phases assigned to the edges of $\mathcal{G}'$ by $\mathbf{H}$ provide a \textbf{\emph{swift phases configuration}} for $\mathcal{G}'$.
\end{definition}

We could not provide general criteria for a graph (with minimal degree $\geq 2$) to admit a swift phases configuration. Nevertheless, we shall prove that many infinite families of graphs admit such a construction. Missing proofs can be found in Appendix \ref{apx:proofs}. \\

\begin{theorem}
\label{thm:evenswift}
If all the vertices of $\mathcal{G}'$ have even degree, then it admits a swift phases configuration.
\end{theorem}

\vspace{2mm}

\begin{theorem}
\label{thm:cubicswift}
A $3$-regular graph $\mathcal{G}'$ admits a swift phases configuration if and only if one of the following equivalent conditions is met:
\begin{enumerate}
    \item $\mathcal{G}'$ has a perfect matching 
    \item $\mathcal{G}'$ can be covered with vertex disjoint cycles (it has a $2$-factor)
\end{enumerate}
\end{theorem}

We recall that a perfect matching of $\mathcal{G}'$ is a $1$-factor, i.e. a partition of all the vertices in pairs such that the vertices in each pair are adjacent in $\mathcal{G}'$ and each vertex appears in just a single pair. Moreover, we say that $\mathcal{G}'$ can be covered with vertex disjoint cycles if there exists a set of cycles that are subgraphs of $\mathcal{G}'$, whose union covers all the vertices of $\mathcal{G}'$ and such that each vertex of $\mathcal{G}'$ belongs to a single cycle. This is also called a $2$-factor of $\mathcal{G}'$. Notice that Theorem \ref{thm:cubicswift} also implies that there are regular graphs not admitting a swift phases configuration: an example is provided by the graph in Fig. \ref{fig:cubicnomatch}.

\begin{figure}[h!]
     \centering
     \includegraphics[scale=0.35]{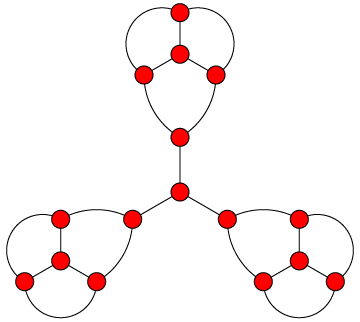}
     \caption{A $3$-regular graph having no perfect matching and therefore no swift phases configuration.}
     \label{fig:cubicnomatch}
 \end{figure} 
 
 \vspace{2mm}
\begin{corol}
A bridgeless $3$-regular graph always admits a swift phases configuration.
\end{corol}
\begin{proof}
Recall that a graph is $k$-edge connected if $k$ is the minimum number of edges to be deleted in order to make the graph disconnected. A non connected graph is bridgless if $k > 1$ for each connected component. Petersen theorem ensures that bridgless $3$-regular graphs admit perfect matching and the thesis follows from Therorem \ref{thm:cubicswift}. 
\end{proof}

\textcolor{red}{Notice that the converse is not true: there exist $3$-regular graphs with perfect matching and with bridges. An example is provided in Fig. \ref{fig:cubicbridge}.}

\begin{figure}[h!]
     \centering
     \includegraphics[scale=0.5]{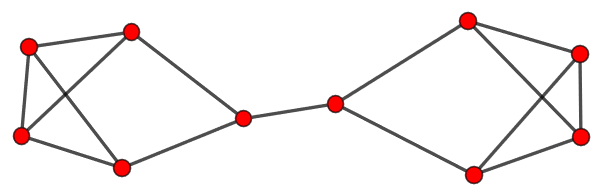}
     \caption{A $3$-regular graph with a bridge and a perfect matching, thus admitting a swift phases configuration.}
     \label{fig:cubicbridge}
 \end{figure}  
 
 \vspace{5mm}
\begin{corol}
An hamiltonian $3$-regular graph always admits a swift phases configuration.
\end{corol}
\begin{proof}
A graph $\mathcal{G}'$ is said to be hamiltonian if it admits a hamiltonian cycle, a closed path passing through each vertex of $\mathcal{G}'$ exactly once. The thesis is then a direct consequence of Theorem \ref{thm:cubicswift} 
\end{proof}

\begin{theorem}
\label{thm:completeswift}
Complete graphs $K_{N}$ and bipartite complete graphs $K_{N_{1}, N_{2}}$ admit a swift phases configuration for all integers $N \geq 3$ and $N_{1}, N_{2} \geq 2$.  
\end{theorem}

\vspace{4mm}
\begin{lema}
Let $\mathcal{G}_{1}$ and $\mathcal{G}_{2}$ be two graphs such that a cone with apex $v$ on $\mathcal{G}_{1}$ (resp. $\mathcal{G}_{2}$) admits a swift chiral quantum walk starting from $v$. Then also the cone over their disjoint union $\mathcal{G} = \mathcal{G}_{1} \bigcup \mathcal{G}_{2}$ admits a swift chiral quantum walk starting from its apex. 
\end{lema}

\vspace{4mm}
\section{Gauge fixing and uniqueness} 
\par{A quasi-gauge transformation for a chiral quantum walk \cite{cqw1,cqw2} is a unitary transformation $\mathbf{U}_{g}$ that is diagonal in the vertex basis and acts on chiral Hamiltonians $\mathbf{H}$ as $\mathbf{H}' = \mathbf{U}_{g} \mathbf{H} \mathbf{U}_{g}^{\dagger}$. They have the important property of preserving site-to-site transport probabilities \cite{TRBIAM21,frigerio21a}. Therefore, as far as we focus on these quantities, the set of chiral Hamiltonians on a graph $\mathcal{G}$ can be partitioned into equivalence classes, such that any two matrices belonging to the same class are linked by a quasi-gauge transformation. Whenever $\mathcal{G}$ is a tree graph (i.e. it has no loops), all chiral adjacency matrices on it belong to the same equivalence class, thus they generate site-to-site transport probabilities identical to those generated by the standard adjacency matrix of the graph (and similarly for chiral Laplacians and, more generically, chiral Hamiltonians for a fixed choice of diagonal elements). In contrast, as soon as $\mathcal{G}$ has at least one loop, there will be an uncountable number of distinct equivalence classes, which can be indexed by the sums of the phases along each loop of the graph, for an initially fixed orientation.} \\

\par{In the light of these observations and before discussing the uniqueness of the swift chiral quantum walk construction on a given graph and for a given starting vertex, it is useful to fix a gauge. Indeed, since quasi-gauge transformations do not affect site-to-site transport probabilities, they will transform a swift chiral quantum walk into another one, without modifying the defining condition Eq.(\ref{eq:swift}).
Let us consider the case of a cone $\mathcal{G}$, with apex labelled by $1$, over a graph $\mathcal{G}'$ of order $N$. The first column $\underline{h}_{1}$ of a chiral Hamiltonian $\tilde{\mathbf{H}}$ on $\mathcal{G}$ will have the following structure $\underline{h}_{1} = ( d_{1} , \underline{h'}_{1} )^{T}$ where $d_{1} \in \reals$ and $\underline{h}'_{1} \in \complex^{N}$ is a vector whose entries are all complex numbers of modulus 1. Then, a gauge-fixing can be imposed by asking that $\underline{h'}_{1} = (1,1,...,1)^{T} \in \reals^{N}$; indeed, in any class of gauge-equivalent chiral Hamiltonians, there will be just one representative fulfilling this condition. With this gauge choice, a chiral adjacency matrix $\tilde{\mathbf{A}}$ on $\mathcal{G}$ will generate a swift chiral quantum walk from $1$ if and only if the corresponding submatrix $\tilde{\mathbf{A}}'$ (which is a chiral adjacency matrix on $\mathcal{G}'$) is such that the sum of each of its rows (and therefore of each column) is zero. }

\vspace{2mm}
\begin{theorem}
Let $\mathcal{G} = ( \mathcal{V} , \mathcal{E} )$ be a simple, connected and undirected graph with $\vert \mathcal{V} \vert = N$ vertices. Let $v \in \mathcal{V}$ be a generic vertex of $\mathcal{G}$ and denote $\mathcal{G}_{v}$ the subgraph of $\mathcal{G}$ obtained by keeping $v$ and all its neighbours, and deleting every other vertex and the respective edges. Then there exists a \emph{swift chiral quantum walk} (of the adjacency type) starting from $v$ in $\mathcal{G}$ if and only if the following conditions hold:
\begin{enumerate}
    \item[\label{cond:a}(a)] There is a swift chiral quantum walk starting from $v$ in the cone $\mathcal{G}_{v}$
    \item[\label{cond:b}(b)] Any vertex $v' \in \mathcal{G}$ which is not a neighbour of $v$ has either $0$ or $\geq 2$ common neighbours with $v$. 
\end{enumerate}
\end{theorem}
\begin{proof}
\textbf{\emph{Sufficiency of conditions.}}
Let us call $N'$ the cardinality of the subset of $\mathcal{V}$ consisting of the neighbors of $v$, so that there are exactly $N- N' - 1$ vertices in $\mathcal{G}$ which are not connected to $v$. Let us enumerate with $1$ the vertex $v$, with progressive integers from $2$ to $N'+1$ the neighbours of $\mathcal{V}$ and with the remaining integers from $N'+2$ to $N$ the rest of the vertices. Then the adjacency matrix of $\mathcal{G}$ will be written in the following block form:

\begin{equation}
\label{eq:Av}
    \mathbf{A} \ = \ \left( \begin{array}{c|ccc|ccc}
    0  & 1 & \textcolor{red}{\cdots} &  \textcolor{red}{1} & 0 & \textcolor{red}{\cdots} & \textcolor{red}{0} \\
    \hline 
    1 & & & & \\
    \textcolor{red}{\vdots} & & \mathbf{A'} & & & \mathbf{B}^\mathrm{T} & \\ 
    \textcolor{red}{1} & & &  & \\
    \hline
    0 & & & &  & &  \\
    \textcolor{red}{\vdots} & & \mathbf{B} & & &   \mathbf{A''} &  \\
    \textcolor{red}{0} & & & &  & &  \\
    \end{array} \right)
\end{equation}
where $\mathbf{B}$ is an $(N-N'-1) \times N'$ matrix (whose entries are either $0$ or $1$ since $\mathbf{A}$ is an adjacency matrix) and $\mathbf{B}^\mathrm{T}$ is its transpose. 
Calling $\underline{e}_{1} = (1,0,0,0...)^T$ the vector associated with vertex $v$, we have:
\begin{equation}
    \mathbf{A} \underline{e}_{1} \ = \ \sqrt{N'} \underline{e}_{v} \ = \ 0 \oplus \mathbf{1}_{N'} \oplus \mathbf{0}_{N-N'-1} \in \complex \oplus \complex^{N'} \oplus \complex^{N-N'-1} \ \simeq \ \complex^{N}
\end{equation}
\textcolor{red}{where $\underline{e}_{v} = \frac{1}{\sqrt{N'}} ( 0, 1, \cdots, 1, 0, \cdots , 0)^{T}$}. By condition (a), we can attach complex phases to the nonzero entries of $\mathbf{A'}$ to get a chiral Hamiltonian $\mathbf{H'}$ reflecting the same graph topology of $\mathbf{A'}$, and such that $\mathbf{H'} \mathbf{1}_{N'} = 0$. As previously discussed, we can always assume a gauge choice that keeps the first row and column of $\mathbf{A}$ as in Eq.(\ref{eq:Av}): if $\mathcal{G}$ is not a cone, this will not fully fix the gauge since additional local phase changes can be made on the vertices not directly connected with $1$.  Moreover, letting $j > N'+1$, by condition (b) we can deduce that the $(j-N'-1)$th row of $\mathbf{B}$ contains either with $N'$ zeroes \emph{or} at least two entries of the row are $1$s. Therefore, by replacing the ones with the roots of unity whenever a row of $\mathbf{B}$ is not the null vector, we can always arrange a matrix of phases $\mathbf{X}$ whose nonzero entries are in the same place as those of $\mathbf{B}$ and which fulfils $\mathbf{X} \mathbf{1}_{N'} = 0$. Combining these facts, we can then define a chiral Hamiltonian:
\begin{equation}
\label{eq:Hv}
    \mathbf{H}_{v} \ = \ \left( \begin{array}{c|ccc|ccc}
    0  & 1 & 1 &  \cdots & 0 & 0 & \cdots \\
    \hline 
    1 & & & & \\
    1 & & \mathbf{H'} & & & \mathbf{X}^\dagger & \\ 
    \vdots & & &  & \\
    \hline
    0 & & & &  & &  \\
    0 & & \mathbf{X} & & &   \mathbf{A''} &  \\
    \vdots & & & &  & &  \\
    \end{array} \right)
\end{equation}
describing the same graph $\mathcal{G}$, which reduces to a $2 \times 2$ real block in the subspace generated by $\{ \underline{e}_{1}, \underline{e}_{v} \}$:
\begin{equation}
     \tilde{\mathbf{H}}_{v} = \left( \begin{array}{cc} 0 & \sqrt{N'} \\
    \sqrt{N'} & 0 \end{array} \right)
\end{equation}
and therefore the chiral dynamics generated by $\mathbf{H}_{v}$ and starting from $v$ is swift. \\

\textbf{\emph{Necessity of conditions.}}\\
Considering Eq.(\ref{eq:condswift}), it is clear that the returning probability can take the form of a squared cosine if and only if the initial state is a superposition of only two energy eigenstates, such that its dynamics is confined to a two-dimensional subspace of the full Hilbert space. Moreover, any chiral Hamiltonian $\mathbf{H}$ (of adjacency type) on the unweighted graph, upon acting on $\underline{e}_{1}$ will yield a flat state with support on the neighbours of vertex $1$ and orthogonal to it, coinciding with the first column $\underline{h}_{1}$ of $\mathbf{H}$. Let $\underline{e}_{E} = \frac{1}{\sqrt{N'}} \underline{h}_{1}$ be the normalized vector corresponding to $\underline{h}_{1}$; then we have:
\begin{equation}
    \begin{aligned}
    & \mathbf{H} \underline{e}_{1} \ = \ \sqrt{N'} \underline{e}_{E} \\
    & \mathbf{H} \underline{e}_{E} \ = \ \sqrt{N'} \underline{e}_{1} \ + \ \underline{e'}_{E} \\
    \end{aligned}
\end{equation}
where in the second equation $\underline{e'}_{E}$ is a generic vector orthogonal to $\underline{e}_{1}$ (and possibly with support also outside the neighbours of vertex $1$)\footnote{Notice that this is valid even if $\underline{e}_{E}$ is a vector with complex entries. }. But the conditions that the dynamics reduces to a two-dimensional subspace and that the return probability coincides with Eq.(\ref{eq:condswift}) immediately imply that $\underline{e'}_{E} = 0$, from which the two conditions of the theorem follow. 
\end{proof}


\section{No-go theorems for Laplacian quantum walks}
As was reminded in the introduction, continuous-time quantum walks generated by the Laplacian and starting from a vertex with maximal degree all have the same dynamics and they are sedentary. If we allow for an oracle, this behaviour can be easily fixed to get a swift (non-chiral) quantum walk. Indeed, let $\mathbf{L}$ be the Laplacian of a graph $\mathcal{G}$ which is a cone over some other graph $\mathcal{G}'$, let us label with $1$ the apex of the cone, with a corresponding state vector $\underline{e}_{1}$, and call $N$ its degree. Then if we define:
\begin{equation}
\label{eq:Loracle}
    \mathbf{\tilde{L}} \ := \  \mathbf{L} - \textcolor{red}{(N-1)} \mathbf{P}_{1}
\end{equation}
where $\mathbf{P}_{1}$ is the orthogonal projector on $\underline{e}_{1}$, it is straightforward to see that in the reduced basis generated from $\underline{e}_{1}$ the matrix $\mathbf{\tilde{L}}$ takes the same form of Eq.(\ref{eq:Hswift}), \textcolor{red}{apart for the addition of an immaterial identity matrix and a sign difference,} therefore the returning probability is that of a swift quantum walk, Eq.(\ref{eq:swift}). Notice that this is equivalent to Grover's Hamiltonian for quantum search on a graph, generalized to the case where just the target vertex has maximal degree\footnote{Usually Grover's algorithm on quantum walks just considers the case of complete graphs.}. \\

However, in the spirit of using chiral quantum walks and the phases degrees of freedom, it is interesting to ask whether the oracle of Eq.(\ref{eq:Loracle}) can be replaced by a suitable phases configuration while keeping the swift dynamics. In general, this is not possible. Indeed, we already proved that to get the swift evolution, the reduced Hamiltonian in the reduced basis of the starting state vector $\underline{e}_{1}$ should be $2 \times 2$, and by enforcing the diagonal elements of the Laplacian and attributing generic phases to the links, this matrix must have the following structure:
\begin{equation}
\label{eq:HredL}
     \tilde{\mathbf{H}}_{k} = \left( \begin{array}{cc} N & \sqrt{N} \\
    \sqrt{N} & k \end{array} \right)
\end{equation}
where we chose the second vector of the reduced basis to be the $\underline{e}_{E} = \frac{1}{\sqrt{N}} ( \underline{h}_{1} - N \underline{e}_{1} )$, where $\underline{h}_{1}$ is the first column of the full chiral Hamiltonian we are considering (with diagonal elements fixed by the Laplacian) and the number $k$ is the only parameter we may try to adjust by choosing a phases configuration. For simplicity, suppose that the graph $\mathcal{G}'$ is $m$-regular. Then clearly $ \vert k \vert \leq 2 m \textcolor{red}{+ 1}$, since the matrix minor of the chiral Hamiltonian pertaining to $\mathcal{G}'$ has all the diagonal elements equal to $m\textcolor{red}{+1}$ and exactly $m$ nonzero off-diagonal entries in every row, each of which is a complex phase. But by looking at the dynamics generated by the reduced Hamiltonian (\ref{eq:HredL}) we get a returning probability of:
\begin{equation}
    p_{1 \to 1}^{\mathbf{\tilde{H}}_{k}} (t) \ = \ 1 - \dfrac{2 N}{(k - N)^2 + 4 N} \left( 1 - \cos{ \sqrt{(k - N)^2 + 4 N} t}  \right)
\end{equation}
which is sedentary unless \textcolor{red}{$(k - N )\sim O (\sqrt{N})$}. Therefore, on cones over $m$-regular graphs whose degree does not grow with the number of vertices, chiral quantum walks of the Laplacian type starting from the apex cannot be swift. 
Despite its generality, this result relies on the assumption that the dynamics reduces to a $2 \times 2$ subspace, which would be needed to achieve the exact swift evolution. In fact, a more general no-go theorem holds.\\

\begin{theorem}
\label{thm:chiralLsedent}
Let $\mathcal{G} = (\mathcal{V} , \mathcal{E} )$ be a cone over a subgraph $\mathcal{G}'$ of order $N$ (then $\vert \mathcal{V} \vert = N+1$) and denote by $\underline{e}_{1}$ the characteristic vector of the apex of the cone, labelled by $1$. Let $d^{>}$ be the largest degree for the vertices of $\mathcal{G}'$. Then, if $d^{>} \sim o(N)$ for large $N$, any chiral quantum walk of the Laplacian type starting from $1$ is sedentary\footnote{In the original definition by Godsil \cite{godsilsedent}, sedentarity requires that the return probability has the asymptotic form $1 - \frac{k}{N}$ for constant $k$ and large order $N$. Here we slightly relax this definition by allowing asymptotic scaling of the form $1 - \frac{k}{N^{\alpha}}$ for $\alpha > 0$, such that it is still true that the return probability is close to $1$ as $N \to \infty$. In particular, Theorem \ref{thm:chiralLsedent} implies sedentarity with $\alpha = 1/2$.  }; more generally, if $\tilde{\mathbf{L}}$ is a chiral Laplacian on $\mathcal{G}$:
\begin{equation}
    p^{\tilde{\mathbf{L}}}_{1 \to 1} (t) \ \geq \ 1 - \dfrac{ \sqrt{N}}{N - 2 d^{>}}
\end{equation}
(assuming $2 d^{>} < N$).
\end{theorem}

\begin{proof}
We will prove this statement via energy conservation, i.e. imposing that the expectation value of $\tilde{\mathbf{L}}$ be the same at the beginning and at the end of the unitary evolution. 
Suppose that at a time $t$ the quantum walk starting from $1$ has evolved into:
\begin{equation}
\label{eq:evole1}
    \underline{e}(t) \ = \ e^{-i t \tilde{\mathbf{L}}} \underline{e}_{1} \ = \ \alpha \underline{e}_{1} \ + \ \beta \underline{f}
\end{equation}
where $\underline{f}$ is a normalized vector with support on $\mathcal{G}'$ ($\underline{f}^{\dagger} \underline{f} = 1$ and $\underline{e}_{1}^{\dagger} \underline{f} = 0$), with $\vert \alpha \vert^{2} + \vert \beta \vert^{2} = 1$.
In the appendix, we prove that if $\underline{f}$ is normalized and with support on $\mathcal{G}'$ having maximum degree $d^{>}$, then:
\begin{equation}
\label{eq:appxlema}
    \underline{f}^{\dagger} \tilde{\mathbf{L}} \underline{f} \ \leq \ 2 d^{>}
\end{equation}
Notice also that:
\begin{equation*}
    \underline{e}^{\dagger} (t) \tilde{\mathbf{L}} \underline{e}(t) \ = \ \underline{e}^{\dagger}_{1} \tilde{\mathbf{L}} \underline{e}_{1} \ = \ N
\end{equation*}
which can be combined with Eq.(\ref{eq:evole1}) to get:
\begin{equation*}
    N = \underline{e}^{\dagger} (t) \tilde{\mathbf{L}} \underline{e}(t) \ \leq \ N \vert \alpha \vert^{2} \ + \ 2 d^{>} \vert \beta \vert^{2} \ + \ \left( \alpha \beta^{*} \underline{f}^{\dagger} \tilde{\mathbf{L}} \underline{e}_{1} + \ h.c. \right)
\end{equation*}
where $h.c.$ stands for the hermitian conjugate of the first term in parenthesis and we invoked Eq.(\ref{eq:appxlema}). Since $\tilde{\mathbf{L}} \underline{e}_{1}$ projected onto $\mathcal{G}'$ is a complex vector with norm $\sqrt{N}$ and $\underline{f}$ is normalized and has support on $\mathcal{G}'$, Cauchy-Schwarz inequality in $\complex^{N}$ implies:
\begin{equation*}
    \vert \underline{f}^{\dagger} \tilde{\mathbf{L}} \underline{e}_{1} \vert \leq \sqrt{N}
\end{equation*}
and consequently, using the fact that $\vert \alpha \beta^{*} + \alpha^{*} \beta \vert \leq 1$:
\begin{equation*}
    \vert  \alpha \beta^{*} \underline{f}^{\dagger} \tilde{\mathbf{L}} \underline{e}_{1} + \ h.c.  \vert \leq \sqrt{N} 
\end{equation*}
In conclusion, we arrived at the following inequality:
\begin{equation*}
    N \ \leq \ N \vert \alpha \vert^{2} \ + \ 2 d^{>} \vert \beta \vert^{2} \ + \ \sqrt{N}
\end{equation*}
which can be rearranged by using $\vert \alpha \vert^{2} + \vert \beta \vert^{2} = 1$ into:
\begin{equation}
    N \ \leq \ (N - 2 d^{>})  \vert \alpha \vert^{2} \ + \ 2 d^{>} \ + \ \sqrt{N}
\end{equation}
If $2 d^{>} < N$, then we get the following bound on $\vert \alpha \vert^{2}$:
\begin{equation}
   \vert \alpha \vert^{2} \ \geq \  1 - \dfrac{\sqrt{N}}{N - 2 d^{>}}
\end{equation}
\end{proof}

\textbf{\emph{Remark.}} The fact that $\mathcal{G}$ is a cone is not strictly necessary for the theorem above to hold. Indeed, it is sufficient that the starting vertex has the largest degree ($N$) and that the degrees of all the other vertices are $o(N)$. In this sense, this result indicates that the diagonal elements of the chiral Laplacians strongly hinder the transport between vertices with highly different degrees, and this issue cannot be overcome by leveraging on chirality.  \\

\section{Conclusions}
Previous results indicate that sedentarity, i.e. a return probability close to $1$ at all times, seems a general property of continuous-time quantum walks starting from a maximal degree vertex. This is a general statement when the generator is the Laplacian, thanks to the universality of the corresponding dynamics \cite{razzoli22}, while it is also true for CTQWs generated by the adjacency matrix when the underlying graph is sufficiently connected. In this paper, we have shown that the sedentarity of adjacency type CTQWs can be amended by considering \emph{chiral} adjacency matrices as generators of the unitary evolution. In particular, a configuration of phases on the edges of a graph can be chosen such that the dynamics starting from the apex of a cone over such a graph will leave the apex in the shortest time possible allowed by unitary evolutions generated by generic Hamiltonians compatible with the graph topology. We termed this construction \emph{swift chiral quantum walks} and the question about its feasibility on generic graphs naturally lead to the nontrivial problem of the existence of a swift phases configuration, which we were able to settle for several families of graphs. By contrast, the sedentarity of Laplacian quantum walks cannot always be resolved by adding phases to the edges; we proved that whenever the initial vertex has a degree which is considerably larger than the degrees of all other vertices in the graph, \emph{any} chiral Laplacian will still generate a sedentary evolution. The relevance of these results partly relies in their generality, since up to now very little was known about tasks that can or cannot be achieved by introducing chirality into CTQWs on general graphs.  
Moreover, Theorem \ref{thm:cubicswift} suggests that deep mathematical questions about graphs could arise from the analytical study of the theory of chiral quantum walks, despite the modest attention they elicited so far. Indeed, necessary and sufficient graph-theoretical conditions for generic odd-degree regular graphs to admit a swift phases configuration are still unknown, let alone for non-regular graphs.

\appendix

\section{Proof of theorems on swift phases configurations}
\label{apx:proofs}
\subsection{Proof of Theorem \ref{thm:evenswift} }

\begin{proof} 
We will prove that one can associate phases of $i$ and $-i$ to all the directed edges of $\mathcal{G}'$ such that the sum of the outgoing phases from any vertex is $0$. First notice that this is equivalent to assigning arrows to the edges: each arrow indicates a phase of $i$ when the edge is seen in the corresponding direction, and a factor of $-i$ for the opposite direction. A solution to our problem then would be a configuration of arrows on all edges such that each vertex has an equal number of outgoing and incoming arrows. To achieve such a configuration, let us start from a vertex, pick a random edge and assign an arrow on it directed towards the corresponding neighbour. Repeat this operation again and again, possibly landing more than once on the same vertex, but with the condition that an edge which has already been assigned an arrow cannot be traversed again. This can always be done, since each time we arrive at a vertex (except for the initial one) the number of arrows already assigned to its edges will be odd (two arrows for each time we visited that vertex before, and one additional arrow for visiting it again at the present stage). Therefore, since its connectivity is even by hypothesis, there will be another edge attached to the vertex without an already assigned arrow. Since the graph is finite and this process can continue each time we land on a vertex which is not the starting one, sooner or later we are bound to reach the initial vertex. At that point, each vertex will have an even number of assigned arrows attached to it, and half of them will be outgoing and the other half will be incoming, since each time we arrived at the vertex with an arrow and we left it with another one. Let us now delete all the edges that have been labelled with an arrow: the remaining graph will still satisfy the initial property that each vertex has an even degree, since we removed an even number of edges from each vertex. Therefore, we can start the process again on the remaining edges, until the whole graph will be covered by arrows fulfilling our condition. \\
\end{proof}

 \textcolor{red}{Notice that a shorter proof can be given by relying on Euler's theorem from graph theory, stating that a connected graph has a closed Eulerian trail (a closed path visiting each vertex exactly once) if and only if all its vertices have even degree. }
\subsection{Proof of theorem \ref{thm:cubicswift} }
\begin{proof}
To prove that the first and second condition are equivalent, consider a $3$-regular graph with a perfect matching: the edges that are not involved in the matching must form cycles, since each vertex will have exactly $2$ of such edges attached to it. Moreover, these cycles must be vertex disjoint for the same reason. On the other hand, if $\mathcal{G}$ can be covered with vertex disjoint cycles, each vertex will have exactly one edge attached to it which is not involved in the cycles, and the collection of these edges must form a perfect matching.\\ 

We will  thus prove the theorem by showing that the \textcolor{red}{first condition is necessary and the second is sufficient} for the existence of a swift phases configuration. First, suppose that $\mathcal{G}$ can be covered with vertex disjoint cycles. For each cycle, we pick a random orientation and we attribute a phase of $e^{\frac{2 i \pi}{3}}$ to all of its edges in the chosen orientation. We then assign a phase of $1$ to all the remaining edges, which form a perfect matching, as we argued before. In this way, the three edges attached to any given vertex will have phases of $1$, $e^{ \frac{2i \pi}{3}}$ and $e^{-\frac{2i \pi}{3}}$, summing to $0$ as requested.\\

Now suppose instead that there exists a swift phases configuration on $\mathcal{G}$ and, for convenience, assume that $\mathcal{G}$ is connected (otherwise the reasoning would apply to each connected component). Pick a random edge, which will have an assigned phase of $e^{i \phi}$ in a given orientation, say from vertex $v_{1}$ to vertex $v_{2}$. Now notice that three complex phases can sum to $0$ if and only if they form an equilateral triangle in the complex plane; therefore the phases associated with the other two edges attached to $v_{1}$ must be $e^{i\phi+\frac{2i\pi}{3}}$ and $e^{i\phi-\frac{2i\pi}{3}}$. The same reasoning applies to all other edges, therefore fixing a phase on a single edge restricts the possible choice of phases on all the other edges to a finite set, which we can take to be $\{  e^{i\phi}, e^{i\phi+\frac{2i\pi}{3}}, e^{i\phi-\frac{2i\pi}{3}} \}$ for an appropriately chosen orientation of each edge. At this point there are two possibilities: 
\begin{itemize}
    \item The set $\{  e^{i\phi}, e^{i\phi+\frac{2i\pi}{3}}, e^{i\phi-\frac{2i\pi}{3}} \}$ is equal (as a set) to $\{1, e^{\frac{2i\pi}{3}},e^{-\frac{2i\pi}{3}} \} $ or to $\{-1, -e^{\frac{2i\pi}{3}},-e^{-\frac{2i\pi}{3}} \} $
    \item The set $\{  e^{i\phi}, e^{i\phi+\frac{2i\pi}{3}}, e^{i\phi-\frac{2i\pi}{3}} \}$ and its complex conjugate $\{  e^{-i\phi}, e^{-i\phi+\frac{2i\pi}{3}}, e^{-i\phi-\frac{2i\pi}{3}} \}$ have no number in common
\end{itemize}
In the first case, the edges with phase $1$ (or $-1$) must form a perfect matching, since each vertex must have exactly one edge with such a phase attached to it. In the second case, all the edges with a phase of $e^{i \phi}$ (in any one of the two orientations) still form a perfect matching: indeed, each vertex will have exactly one edge with phase $e^{i \phi}$ or $e^{- i \phi}$ attached to it (and it cannot have both $e^{i \phi}$  and $e^{-i \phi}$). In either case, we conclude that $\mathcal{G}$ must admit a perfect matching.
\end{proof}

\subsection{Proof of Theorem \ref{thm:completeswift}} 

\begin{proof}
An explicit construction when $\mathcal{G}'$ is a complete graph was implicitly provided in \cite{frigerio21b}.
A complete bipartite graph $K_{N_{1},N_{2}}$ is a graph with two sets of vertices $\mathcal{V}_{A}$ with cardinality $N_{1}$ and $\mathcal{V}_{B}$ with cardinality $N_{2}$, such that no edge connects vertices from the same set and every vertex in $\mathcal{V}_{A}$ is connected with every vertex in $\mathcal{V}_{B}$. The adjacency matrix of such a graph can be written in the following form:
\begin{equation}
    \mathbf{A}_{K_{N_{1},N_{2}}} = \left( \begin{array}{cc} \mathbf{0}_{N_{1} \times N_{1}} & \mathbf{B} \\
    \mathbf{B}^{T} & \mathbf{0}_{N_{2} \times N_{2}} \end{array} \right)
\end{equation}
where $\mathbf{B}$ is an $N_{1} \times N_{2}$ matrix filled with $1$s. A chiral adjacency matrix fulfilling Eq.(\ref{eq:condphases}) can be constructed by replacing $\mathbf{B}$ with a matrix of phases $\tilde{\mathbf{B}}$ with entries $\tilde{\mathbf{B}}_{jk} = e^{ i \frac{2 \pi j}{ N_{1} } } e^{i \frac{ 2 \pi k}{ N_{2} }} $, since $\sum_{j=1}^{N_{1}} \tilde{\mathbf{B}}_{jk} = \sum_{k=1}^{N_{2}} \tilde{\mathbf{B}}_{jk} = 0$.
\end{proof}

\section{A result on the spectrum of chiral adjacency matrices}

\begin{theorem}
Let $\mathcal{G}$ be a graph and let $\tilde{\mathbf{A}}$ be a chiral adjacency matrix on $\mathcal{G}$. Let $d^{>}$ be the largest degree among the vertices of $\mathcal{G}$. Then if $\lambda_{A}$ is an eigenvalue of $\tilde{\mathbf{A}}$:
\begin{equation}
    \vert \lambda_{A} \vert \leq d^{>}
\end{equation}
Analogously, if $\tilde{\mathbf{L}}$ is a chiral Laplacian on $\mathcal{G}$ and $\lambda_{L}$ \textcolor{red}{is} one of its eigenvalues:
\begin{equation}
    \vert \lambda_{L} \vert \leq 2 d^{>}
\end{equation}
\end{theorem}

\begin{proof}
Let $\underline{x}$ be an eigenvector of $\tilde{\mathbf{A}}$ with eigenvalue $\lambda$ and let $\vert x_{j_{0}} \textcolor{red}{\vert} = \sup_{j}  \vert x_{j} \vert $ be the largest entry of $\underline{x}$ (in absolute value). Then:
\begin{equation}
    \vert \lambda \vert \vert x_{j_{0}} \vert \ = \ \vert ( \tilde{\mathbf{A}} \underline{x} )_{j_{0}} \vert \ \leq \ \sum_{k} \vert \tilde{\mathbf{A}}_{j_{0} k} x_{k} \vert \ \leq \ \vert x_{j_{0}}  \vert \sum_{k} \vert \tilde{\mathbf{A}}_{j_{0} k} \vert \ \leq \  \vert x_{j_{0}}  \vert d^{>}
\end{equation}
from which the thesis follows\footnote{Notice that $| x_{j_{0}} | > 0$ and it cannot vanish since an eigenvector is never the null vector.}. To derive the above inequalities, we used the fact that the entries of $\tilde{\mathbf{A}}$ are complex numbers of modulus $1$ and that the maximum number of nonzero entries in each row is $d^{>}$, by definition. The same reasoning applied to a chiral Laplacian gives the other bound, since $\sum_{k} \vert \tilde{\mathbf{L}}_{j_{0} k} \vert$ is upper bounded by $2 d^{>}$, because of the diagonal contribution.
\end{proof}

\begin{corol}
If $\underline{f}$ is a normalized vector with support on $\mathcal{G}$, then for any chiral adjacency matrix $\tilde{\mathbf{A}}$ and any chiral Laplacian $\tilde{\mathbf{L}}$:
\begin{equation}
    \vert \underline{f}^{\dagger} \tilde{\mathbf{A}} \underline{f} \vert \ \leq \ d^{>} \ \ , \ \ \ \ \vert \underline{f}^{\dagger} \tilde{\mathbf{L}} \underline{f} \vert \ \leq \ 2 d^{>} 
\end{equation}
\end{corol}

\bibliographystyle{iopart-num}
\bibliography{nonsedcqw}

\end{document}